\documentclass{hotnets2020}
\usepackage{blindtext}
\usepackage{url}
\usepackage{times} 
\usepackage{hyperref}
\usepackage{makecell}
\usepackage{subcaption}
\usepackage{multirow}

\usepackage{algorithmic}
\usepackage[linesnumbered,lined,ruled,commentsnumbered]{algorithm2e}
\usepackage{amsmath}
\newtheorem{lemma}{Lemma}
\usepackage{graphicx}
\usepackage{float}

\usepackage{authblk}
\usepackage{balance}
\usepackage{ragged2e}

\hypersetup{pdfstartview=FitH,pdfpagelayout=SinglePage}

\usepackage{xcolor}

\newcommand{\sys}{{DCFIT}\xspace}

\begin{document}
\title{\sys: Initial Trigger-Based PFC Deadlock Detection in the Data Plane}

\author{\large Xinyu Crystal Wu, T.S. Eugene Ng}
\affil{\large Rice University}


\maketitle

\begin{abstract}
Recent data center applications rely on lossless networks to achieve high network performance. Lossless networks, however, can suffer from in-network deadlocks induced by hop-by-hop flow control protocols like PFC. Once deadlocks occur, large parts of the network could be blocked. Existing solutions mainly center on a deadlock avoidance strategy; unfortunately, they are not foolproof. Thus, deadlock detection is a {\em necessary} last resort. 
In this paper, we propose \sys, a new mechanism performed entirely in the data plane to detect and solve deadlocks for arbitrary network topologies and routing protocols. Unique to \sys is the use of deadlock initial triggers, which contribute to efficient deadlock detection and deadlock recurrence prevention. 
Preliminary results indicate that \sys can detect deadlocks quickly with minimal overhead and mitigate the recurrence of the same deadlocks effectively. This work does not raise any ethical issues.


\end{abstract}

\section{introduction}
\label{section:intro}

Driven by demand for ultra-low latency, high throughput network applications with low CPU overhead, lossless networks are widely deployed in modern data centers~\cite{zhu2015congestion, zhu2015packet}. One typical implementation of such networks supported by Remote Direct Memory Access (RDMA) built upon Ethernet relies on Priority-based flow control (PFC) to prevent buffer overflow~\cite{ieee}. With the PFC mechanism, packet loss can be avoided by pausing the immediate upstream switch.

Although effective, PFC can induce an intrinsic problem: deadlocks caused by cyclic buffer dependency (CBD), where no packets in the cycle can be propagated. 
Once deadlocks occur, PFC pause frames could spread to significant parts of the fabric, causing a large percentage of flows to stop transmission. In the worst case, all ports along all paths would be ceased and the whole network would be blocked.

Many large cloud providers have confirmed that deadlocks occur in practice frequently~\cite{stephens2014practical, qian2019gentle, halperin2011augmenting}. Deadlocks could happen when routing rules form a loop, but it is not a unique product of routing loops---recent work has shown that even for tree-based topology with up-down loop-free routing, deadlocks could still occur due to situations such as link failures~\cite{wu2012netpilot, lou2020understanding, tan2019netbouncer}, transient loops~\cite{jin2014dynamic}, complex network updates~\cite{jin2014dynamic, forster2016consistent}, port flaps~\cite{liu2013f10}, and misconfigurations~\cite{zhu2015packet, kakarla2020finding, hu2017tagger}. 
Furthermore, deadlocks do not recover automatically even after the problems (e.g. transient loop) that caused the deadlock formation have been fixed ~\cite{hu2016deadlocks}.
Although recent works have proposed to replace PFC with a smarter retransmission mechanism over lossy networks~\cite{mittal2018revisiting, mittal2018towards, mittalgood}, many services like inter-process communication still require a lossless guarantee to obtain optimal performance.  Furthermore, short-lived flows, commonly seen in data center networks~\cite{benson2010network, zats2015fastlane}, are more sensitive to packet loss---any packet retransmission would significantly increase their flow completion time~\cite{alizadeh2010data, zats2012detail}. Therefore, a lossless network remains very desirable and resolving deadlock induced by PFC is still an important issue.

Existing approaches to resolving deadlock fall into two main categories: avoidance and detection/recovery. Most recent research efforts focus on deadlock avoidance, but none of them is foolproof. For example, some approaches require special buffer management techniques to support multiple priority classes. However, commodity switches in modern data centers can only support two or three lossless priorities in practice~\cite{guo2016rdma, hu2017tagger}. Some schemes rely on limiting the rate of data transmission and thus avoid reaching the threshold of generating pause frames~\cite{qian2019gentle}. However, precise and fine-grained control of the sending rate may not always be guaranteed by switch implementations. No avoidance method can absolutely prevent deadlocks. 
As a result, an efficient and accurate deadlock detection method is a necessary safeguard. 

Detecting deadlocks in distributed networks is difficult as it usually relies on indirect inspection to detect suspected ports. Existing solutions for deadlock detection rely on centralized controllers or switch software control planes to detect and break the deadlock. However, both are inherently too slow to deal with deadlocks. Besides, existing solutions ignore finding the root cause of a deadlock, so even if the current deadlock is addressed, there is no guarantee that the same deadlock would not immediately reoccur again.

In this paper, we propose \sys---a deadlock detection mechanism entirely performed in the data plane. \sys works for arbitrary network topologies and routing protocols, and can react quickly after the deadlock formation and provide the root cause information for resolving the deadlock. Rather than continually monitor the throughput and queue occupancy of each switch port which incurs unnecessary overhead, \sys only triggers the detecting process when pause events happen. 
The metadata for deadlock detection is piggybacked onto pause frames or synthesized packets according to different situations. \sys provides a basis to analyze the initial trigger of the current deadlock, which helps to address deadlock recurrence.
Our preliminary results show that \sys can detect deadlocks accurately with low latency and minimal switch memory consumption. Furthermore, resolving the initial trigger can effectively prevent the same deadlock from recurring.

\section{Motivation}
\label{section:motivation}
\vspace{-2mm}
\subsection{Deadlock Avoidance - Not Foolproof}

\textbf{Restricted routing.}
The most common solutions for deadlock avoidance are to restrict routing paths and avoid the formulation of CBD~\cite{duato2001general, skeie2002layered, wu2003fault, sancho2004effective, blazewicz1994optimal, domke2011deadlock}.
However, routing restrictions not only waste link bandwidth and reduce throughput~\cite{hu2016deadlocks}, but also are incompatible with some topologies~\cite{dally1988deadlock, park1996generic} and routing protocols such as OSPF and BGP~\cite{stephens2016deadlock, stephens2014practical}. Furthermore, when some links are down, rerouting could still create CBD and lead to deadlocks~\cite{wu2012netpilot, liu2013f10}.

\textbf{Buffer management.}
Another method is to assign packets different priorities hop-by-hop and put packets into different buffers accordingly~\cite{lee2005prevention, underwood2011unified}.
The required number of priorities is determined by the longest path in the network, which increases with the network scale. However, today's commodity switches can only support two to three lossless priorities in practice~\cite{guo2016rdma, hu2017tagger}, which is insufficient.

\textbf{PFC pause frame restrictions.}
Recent proposed congestion control protocols~\cite{zhu2015congestion,mittal2015timely, li2019hpcc, cheng2020re, geng2019p4qcn} can reduce the possibility of pause. Also, operators can limit pause frame propagation by assigning different PFC thresholds to ports and switches based on their positions in the topology~\cite{hu2016deadlocks}. Although these methods can reduce the possibility of a deadlock, they cannot absolutely prevent deadlocks.

\textbf{Pre-configured transmission.}
Deadlock avoidance can also rely on configuring the network in advance and adjusting the configuration dynamically on the fly. For example, Tagger~\cite{hu2017tagger} pre-defines expected lossless paths and configures pre-generated match-action rules to avoid deadlock. However, it involves human intervention and complex network configuration that is pointed out to be error-prone~\cite{birkner2020config2spec, fogel2015general, beckett2017general}.

\textbf{Rate limiting.}
Rate limiting is used to break the necessary condition---\textit{hold and wait}---for the deadlock~\cite{qian2019gentle}. Nonetheless, in order to limit the port rate to specific values and to periodically adjust the port rate according to the queue length, highly precise control are required. However, precise and fine-grained control of the sending rate may not be always guaranteed by switch implementations.

\textbf{Summary.}
Existing deadlock avoidance approaches address the problem to some extent, but they are not foolproof. Therefore, deadlock detection is an important and necessary fail-safe.

\subsection{Existing Detection Solutions Fall Short}
Existing deadlock detection solutions rely on a centralized controller or switch local control planes~\cite{dally1988deadlock, anjan1995efficient, martinez1997software, lopez1998very, shpiner2016unlocking} to query port states and detect deadlocks. After detecting deadlocks, the control plane computes and installs rerouting or draining strategies to switches. The programmability of central controllers or local control planes enables flexible detection and recovery policies. However, inherent delays between data planes and control planes together with the software delays of control plane applications make these solutions unable to response to deadlocks fast enough.

In addition, existing solutions detect deadlocks by proactively monitoring for blocked ports. Concretely, if the throughput of a port is zero while the corresponding queue length is non-zero, the port is regarded as a suspected port that can form deadlocks. However, the overhead of proactive monitoring is very high as it requires the periodic inspection of all ports of all switches in a network. In a normal network, most of time, the inspection will find nothing wrong.

Furthermore, current deadlock detection solutions are unable to eliminate the root cause of the detected deadlock. Therefore, even if the deadlock can be broken and the traffic flow can recover temporarily, none of them is able to prevent the same deadlock from reappearing again.

\subsection{New Opportunity for Deadlock Detection}
For each shortcoming of existing solutions, we propose a new alternative design strategy in \sys.

\textbf{Detecting deadlocks in the data plane.}
Recent proposed programmable data plane provides a new opportunity for deadlock detection~\cite{bansal2012openradio, arashloo2016snap, hancock2016hyper4}. First, programmable parsers and deparsers enable us to customize packet headers and protocols. Metadata for deadlock detection can thus be piggybacked onto pause frames.
Second, the provided stateful memory and ALUs make it possible to maintain state information directly in the data plane. A deadlock detection algorithm performed entirely in the data plane reduces the high overhead introduced by interacting with the switch control plane. Finally, once compiled, data plane programs are guaranteed to run at line speed, which allows us to react quickly when the deadlock is formed.

\textbf{Detecting deadlocks reactively.} 
Rather than periodically checking the status of switch ports, \sys starts the deadlock detection process only when an initial pause event happens. The detecting process follows the direction of pause frame propagation, which greatly reduces the overhead. Besides, being triggered by the initial pause event, \sys is able to detect deadlocks almost immediately once they have been formed. 

\textbf{Preventing recurrence of the same deadlock.}
Breaking the deadlock by adjusting a switch on the CBD might be insufficient, as a chain of similar pause events could be generated along switches on the path again and cause the same deadlock. Experimental observations from existing work~\cite{qian2019gentle} show that frequent pause on upstream ports is the root cause of \textit{hold and wait}, which then forms deadlocks. Therefore, figuring out the origin of frequent pause and mitigating the following pause events help to resolve deadlocks. The Ethernet node that initially sends the pause frame and leads to the deadlock can be regarded as the initial trigger. Identifying and dealing with the initial trigger can prevent the repeated formulations of the same deadlock. 
\section{System design}
\label{section:system}

\sys leverages the programmable data plane to detect deadlocks and identify initial trigger nodes. Detection processes are based on different locations of the initial trigger - on the loop or out of the loop, respectively. A port-based data structure is used to keep track of causal relationships between pause events generated at different switches. \sys attaches metadata for deadlock detection to pause frames or synthesized packets. Once a deadlock is detected, the initial trigger provides clues to mitigate potential pause events later and hence prevent the same deadlock from recurring.

\sys focuses on two main principles---1) spatial: a chain of pause events triggered hop by hop (causality-chain) forms a loop (causality-loop) 2) temporal: all nodes on the causality-loop are paused simultaneously, indicating a real deadlock rather than just a CBD scenario. \sys makes no assumptions about the topologies and routing algorithms in use. 

\subsection{Identifying the Initial Trigger}
The initial trigger, which can be a server or a switch, is at the beginning position of the causality chain. A server generating a pause frame is immediately identified as an initial trigger since it is the destination of flows in the network. A switch is identified as an initial trigger if it has not received any pause frame from the corresponding downstream when generating its own pause frame. 

When a pause frame is triggered at a switch port, it might be caused by previously received pause frames at other ports. It is because some of the traffic from one switch port is destined to other ports that have already been paused. Therefore, when an ingress port generates a pause frame due to the congestion at a corresponding egress port, it checks whether the egress port is paused or not. An initial trigger is identified when the egress port is not paused.

\subsection{General Primitives for Deadlock Detection}
\label{primitives}
Before discussing deadlock detection based on the location of initial triggers, we present primitives for solving two basic aspects of deadlock detection. Such primitives can then be applied to cover different scenarios, including the initial trigger on the deadlock loop or out of the deadlock loop. A port-based causality data structure is maintained in the data plane to determine how the metadata for deadlock detection is forwarded. The symbols used in the following are displayed in Table~\ref{table:symbols}.

\begin{figure}[t]
    \centering
    \includegraphics[width=0.45\textwidth]{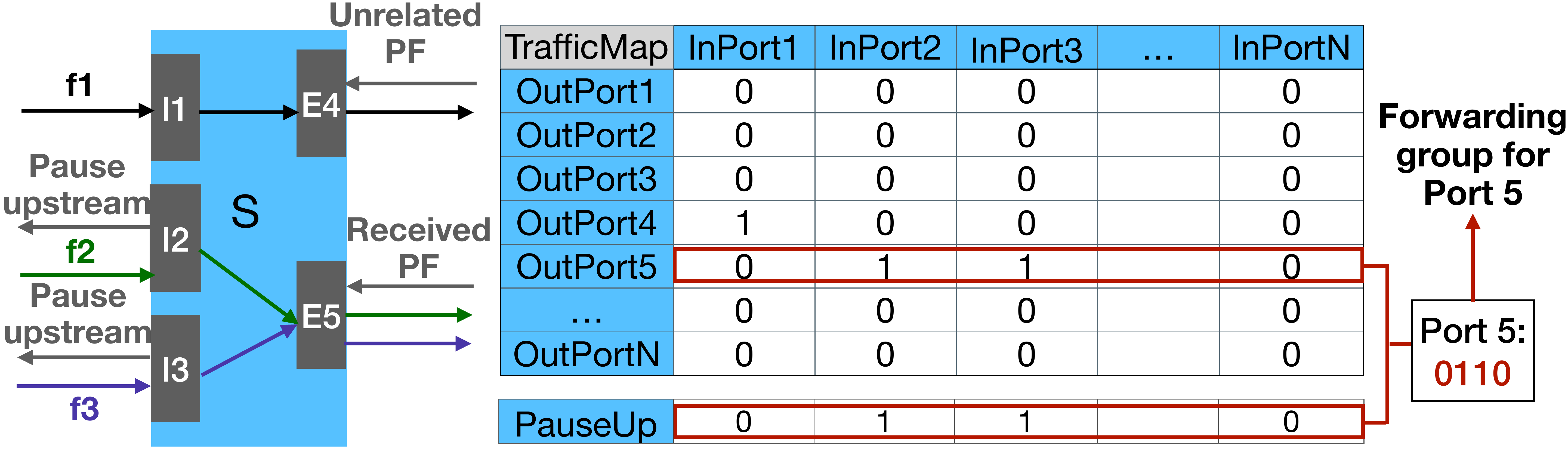}
    \vspace{-0.1in}
    \caption{Port-based causality data structure}
    \label{fig:causality}
    \vspace{-0.25in}
\end{figure}

\textbf{Port-based causality data structure:}
The port-based causality data structure maintains the causality relationship between different switch ports, as shown in Figure~\ref{fig:causality}. 
For a switch with N ports, each port maintains a bit-map of size N to track all the relevant ports that send traffic to its egress queue, which we call a traffic mapping. Each bit indicates whether there are active packets transmitting from a certain ingress port to a certain egress port. In the above example, the egress queue of port $E4$ is occupied by packets of flow $f1$ from ingress port $I1$, making the corresponding position in the traffic mapping be set to $1$. It will be cleared to $0$ when no active packet is in the egress queue. Similarly, for port $E5$, the bits for $I2$ and $I3$ are set to $1$.

Each port also maintains a bit that represents whether the corresponding port currently pause the upstream switch. In the example, port $I2$ and $I3$ have already sent pause frames, the corresponding values are set to $1$. When receiving a new pause frame, the switch would query the traffic mapping as well as the port state to 
determine the subsequent forwarding groups for the metadata used for deadlock detection. Notice that the direction of pause frames is from downstream to upstream, opposite to normal traffic flow. If $E5$ receives a pause frame, the switch will traverse the bit-map for $E5$ and the port states of all ports. Since $I2$ and $I3$ both have causality with $E5$ and currently pause the upstream, the metadata for deadlock detection is forwarded to $I2$ and $I3$.

\begin{table}[H]
    \scriptsize
    \begin{tabular}{ll}
    \hline
    Symbol & Meaning \\
    \hline \hline 
    $S_{tri} $    & Node ID of the initial trigger      \\
    $P_{tri}$    & Port ID of the initial trigger that sends pause frame\\
    $S_{gen-ini}$    & Node ID of the generic initiator      \\
    $P_{gen-ini}$    & Port ID of the generic initiator that sends pause frame\\
    $Seq_{id}$  & Sequence number of checking message sent by $S_{gen-ini}$\\
    $S_{cur}$   & Switch ID of the current switch        \\
    $\xi_{p}$ & Set of causal ports sending traffic to port $p$ at current switch\\
    $\delta_{p}$ & Set of ports that pause the upstream and be causal with port $p$\\
    $r_p$      & RESUME tag for port $p$ of the current switch\\
    \hline
    \end{tabular}
    \vspace{-0.05in}
    \caption {Meaning of symbols used in this paper. }
    \vspace{-0.15in}
    \label{table:symbols}
\end{table}

\textbf{Causality-loop primitive:}  
The causality-loop can be determined when the causality chain of pause frames has visited the same port of the same switch twice. 
To detect the causality-loop with minimal memory overhead, \sys uses the switch suspecting a causality-loop could be formed, which we called a generic initiator. Messages used for tracing the causality-chains are called checking messages. The packet header of a checking message is extended to record the unique ID \{$S_{gen-ini}$, $P_{gen-ini}$\} of the generic initiator, as well as the $Seq_{id}$ sent by the generic initiator. The $Seq_{id}$ represents a unique episode of causality-loop detection. It is used to decompose different causality-chains from the same generic initiator when resume and pause frames alternate.

Given a generic initiator for the causality-loop detection, the following steps are followed:   

\begin{itemize}
    \itemsep0em
    \vspace{-1mm}
    \item When the generic initiator sends out a checking message, it attaches its own \{$S_{gen-ini}$, $P_{gen-ini}$, $Seq_{id}$\} to the packet header. The selection of the generic initiator and when checking messages are sent are different for different use cases based on the location of the initial trigger (details in sections~\ref{Initial Trigger on Loop} and~\ref{Initial Trigger out of Loop}).
    \item When receiving a checking message, non-generic initiator switches parse and store the received \{$S_{gen-ini}$, $P_{gen-ini}$, $Seq_{id}$\} in the data plane at the receive port, which are then used for generating the next checking message. The $Seq_{id}$ is updated when receiving a new checking message from the generic initiator.
    The stored generic initiator info at a port is removed after receiving a resume frame.
    \item When port $p$ receives a checking message, non-generic initiator switches query the port-based causality data structure to obtain corresponding ports $\delta_{p}$ that currently pause the upstream and have causality with port $p$.
    If $\delta_{p}$\,= $\emptyset$, the switch will drop the current checking message and wait for the generation of the next checking message (see next bullet). Otherwise, the switch will forward the checking message to all ports in $\delta_{p}$.
    \item A non-generic initiator switch generates a new checking message when a pause frame is triggered by congestion on an egress port $p$.
    This new checking message will carry the corresponding \{$S_{gen-ini}$, $P_{gen-ini}$, $Seq_{id}$\} stored at port $p$.
    \item When a switch $S_{cur}$ receives a checking message from port $p$ whose $S_{gen-ini}$ is $S_{cur}$, $Seq_{id}$ is the latest, and $P_{gen-ini}$ belongs to $\xi_{p}$, the causality-chain has passed the same port of the same switch again. A causality-loop is determined and a deadlock is potentially formed.
\end{itemize}

\begin{lemma} 
If deadlock exists and $S_{gen-ini}$ is on the CBD, the causality-loop must be detected by $S_{gen-ini}$, when the received $S_{gen-ini}=S_{cur}$,
$P_{gen-ini} \in \, \xi_{p}$ 
where $p$ is the port receiving the checking message, and $Seq_{id}$ indicates the same episode.
\label{lemma:on loop}
\end{lemma}
\begin{proof}
For any given deadlock, a causality-loop must be implied. The checking message spreads with the causality-chain must also follow the causality-loop and return back to $S_{gen-ini}$. At this moment, $S_{cur}$ = \,$S_{gen-ini}$. The next step is to make sure that the checking message received at port $p$ has causalities with ports in $\xi_{p}$ connected to the upstream switches. If any port in $\xi_{p}$ is equal to the received $P_{gen-ini}$, $S_{cur}$ must have already sent out the checking message with the same $P_{gen-ini}$. It indicates there must be a repetitive ID and the causality-chain comes back to the beginning position of the causality-loop. If the recently sent $P_{gen-ini}$ has the same $Seq_{id}$ as the received $P_{gen-ini}$, the causality-chain must have traversed one port of one switch twice in the same episode, which 
verifies the causality-loop. 
\end{proof}
\vspace{-0.05in}
\textbf{Temporal consistency primitive:} 
The temporal consistency primitive is triggered after the causality-loop of pause behaviors is detected by the above mechanism. It is achieved by another round of further check also from the direction of pause frames. The paused ports on the causality-chain might be resumed during the process of causality-loop detection. Therefore, even if a causality-loop is detected, deadlock may not actually exist. Without this primitive, a deadlock could be declared mistakenly.

We leverage a strategy that checks if every node of the causality-loop is still paused ever since the initial pause event has been triggered. Each episode is determined by the $Seq_{id}$ of the generic initiator. Once a pause event occurs, it will hold until a resume frame is received. Each port on the switch maintains a RESUME tag $r_p$ representing whether the pause is resumed during the current detection process. It is set to zero in each episode when the corresponding port is paused, and updated to $1$ when receiving a resume frame. 
A synthesized temporary consistency check packet carrying \{$S_{gen-ini}$, $P_{gen-ini}$, $Seq_{id}$\} is sent passing through the causality-loop, which is obtained by querying the traffic mapping data structure and comparing the corresponding generic initiator ID. The synthesized packet is forwarded to ports that belong to the causal ports $\xi_{p}$ and have sent PFC packets with the same \{$S_{gen-ini}$, $P_{gen-ini}$\}. Each switch port that receives the temporal consistency check packet would check the corresponding RESUME tag and the stored $Seq_{id}$. 
A deadlock is determined if all switches on the causality-loop maintain false RESUME tags and the same $Seq_{id}$ as that during causality-loop detection.

\begin{lemma}
A deadlock is determined if and only if for each switch and each port $p$ along the causality-loop, $r_p$ = \,$0$ and $Seq_{id}$ is the same as the one used to determine the causality-loop.
\label{lemma:on loop3}
\end{lemma}

\begin{proof}
1) If $r_p$ = \,$0$ holds in the whole episode represented by the same $Seq_{id}$: no resume frame is transmitted during the current detecting episode. All ports of the causality-loop is paused through the entire duration and a deadlock is determined. 
2) If $r_p$ = \,$1$ at some switch: port $p$ is resumed and the status is not changed in the current episode. At least one of the nodes on the causality-loop is resumed so that pause events are unable to form a deadlock.
3) If $Seq_{id}$ is changed at some switch, there must be a new pause frame following a resume frame. However, this must start a new independent episode of detection, which means $r_p$ = \,$1$ followed by $r_p$ = \,$0$ cannot happen in the same episode. Therefore, even if $r_p$ = \,$0$ holds for all switches on the loop, they are involved in different episodes so that the deadlock is not determined in the current episode. It is the job of the new detection episode to determine whether a deadlock is formed.
\end{proof}

\subsection{Initial Trigger on Loop}
\label{Initial Trigger on Loop}
When the initial trigger switch is part of the deadlock CBD loop, the deadlock can be detected from the loop itself. Figure~\ref{fig:trigger on loop} shows an example of this case in the Clos network topology. The principle for deadlock detection is identical to the general primitives previously described in Section~\ref{primitives}. 

\begin{figure}[t!]
    \centering
    \includegraphics[height=1.6in]{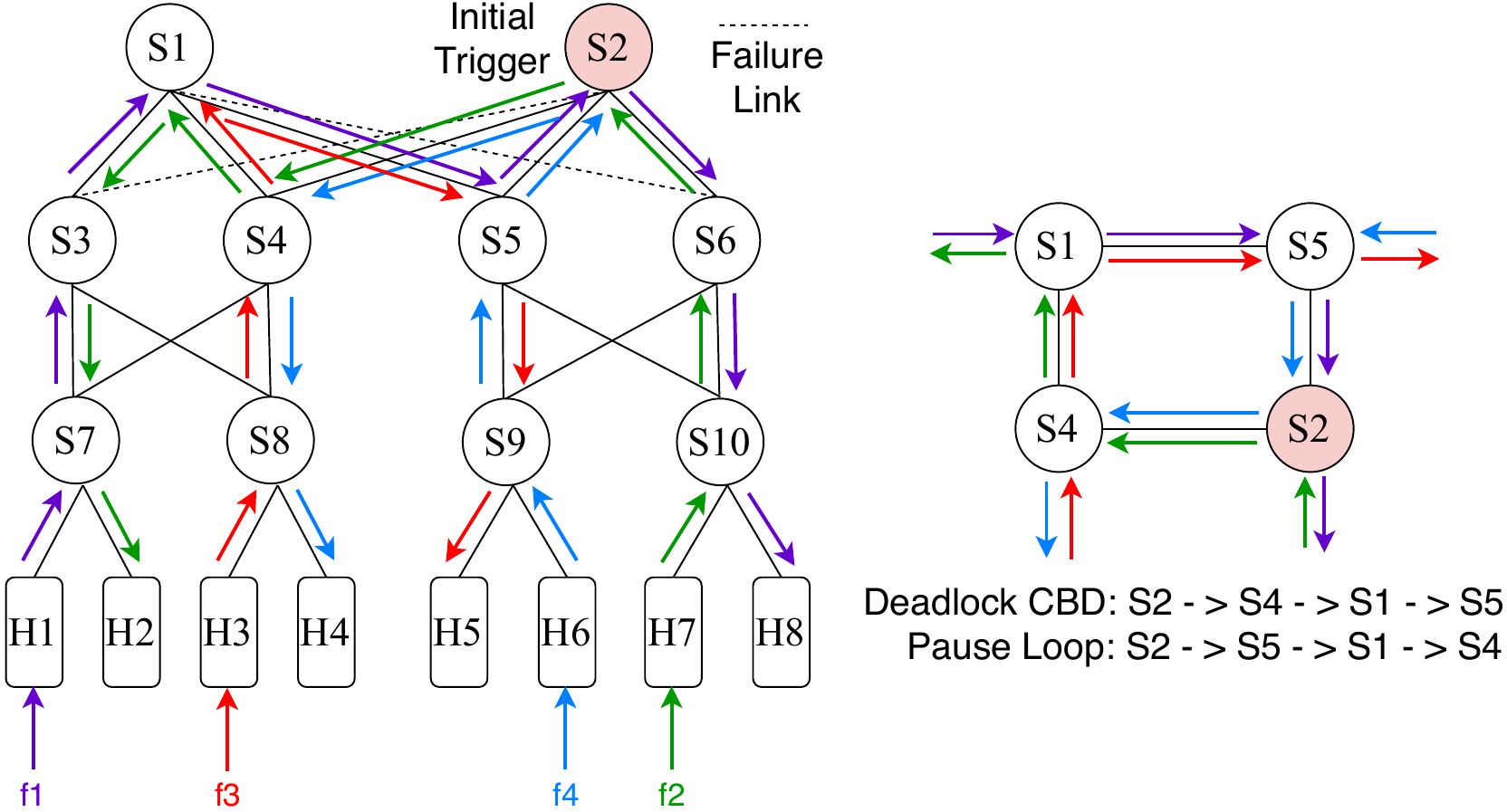}
    \vspace{-0.1in}
    \caption{Example of deadlock in the Clos network topology when initial trigger is on the loop}
    \label{fig:trigger on loop}
    \vspace{-0.1in}
\end{figure}

\textbf{Generic initiator selection:}
In this case, the initial trigger node is regarded as the generic initiator since it is at the beginning of the causality-chain and the causality-chain itself forms the CBD cycle. 

\textbf{Checking message propagation:}
Whether it is the initial trigger or non-initial trigger switch sending out a PFC pause frame, the current checking message \{$S_{tri}$, $P_{tri}$, $Seq_{id}$\} is piggybacked onto the PFC pause frames. When a non-initial trigger switch detects that the next hop of the causality-chain has already been paused, no new pause frame can be generated. The checking message \{$S_{tri}$, $P_{tri}$, $Seq_{id}$\} is then piggybacked onto a synthesized packet with a different priority. 
Based on the causality-loop primitive, if a deadlock exists, the causality-loop must be detected by $S_{tri}$.

\textbf{Temporal consistency guarantee:}
The temporal consistency primitive can take effect in this case by choosing the initial trigger as the start of the temporal consistency check. Deadlock is determined when the temporal consistency primitive holds.

\subsection{Initial Trigger out of Loop}
\label{Initial Trigger out of Loop}
Some practical deadlock scenarios are affected by pause events sent from switches out of the loop. 
As shown in Figure~\ref{fig:trigger out of loop}, malicious flow $f5$ with constant high sending rate causes a pause frame initially to be sent from $S10$, leading to the final deadlock loop between $S2$, $S5$, $S1$ and $S4$. Even if the deadlock can be broken from one of the switches on the loop, without solving continuous pause events from the initial trigger switch $S10$, $S6$ and $S2$ can be paused again and finally lead to a repeated formation of the same deadlock.

We observe that a signal of this case is that a middle switch receives multiple pause frames with the same \{$S_{tri}$, $P_{tri}$\} from different ports. The $Seq_{id}$ does not need to be the same as the initial trigger may alternately send pause and resume frames, which update the $Seq_{id}$ received from outside the loop. However, the detection of this case only cares about the same $Seq_{id}$ on the causality-loop. Therefore, even if the two pause frames with the same \{$S_{tri}$, $P_{tri}$\} received from different ports have different $Seq_{id}$, the detection process can still be triggered.
In the example of Figure~\ref{fig:trigger out of loop}, this phenomenon is detected at $S2$ when it receives two pause frames with the same \{$S_{tri}$, $P_{tri}$\} from $S6$ and $S4$.

\vspace{-0.05in}
\begin{lemma}
If deadlock exists and $S_{tri}$ is out of the causality loop, one switch (called the middle switch) must receive at least two pause frames with the same \{$S_{tri}$, $P_{tri}$\} from different ports.
\label{lemma:out of loop}
\end{lemma}

\begin{proof}
Once the causality-loop is formed and $S_{tri}$ is out of the loop, there must be a switch at the junction between the inside and outside of the loop. This switch has received pause frames from at least two directions. One is from the outside downstream switch, the other one is from the switch on the causality-loop. As both pause frames can be traced back to the $S_{tri}$, the received \{$S_{tri}$, $P_{tri}$\} from different ports must be the same. Notice that since pause and resume frames are possibly alternating continuously, the same \{$S_{tri}$, $P_{tri}$\} received from the same port are not considered.
\end{proof}

\begin{figure}[t!]
    \centering
    \includegraphics[height=1.6in]{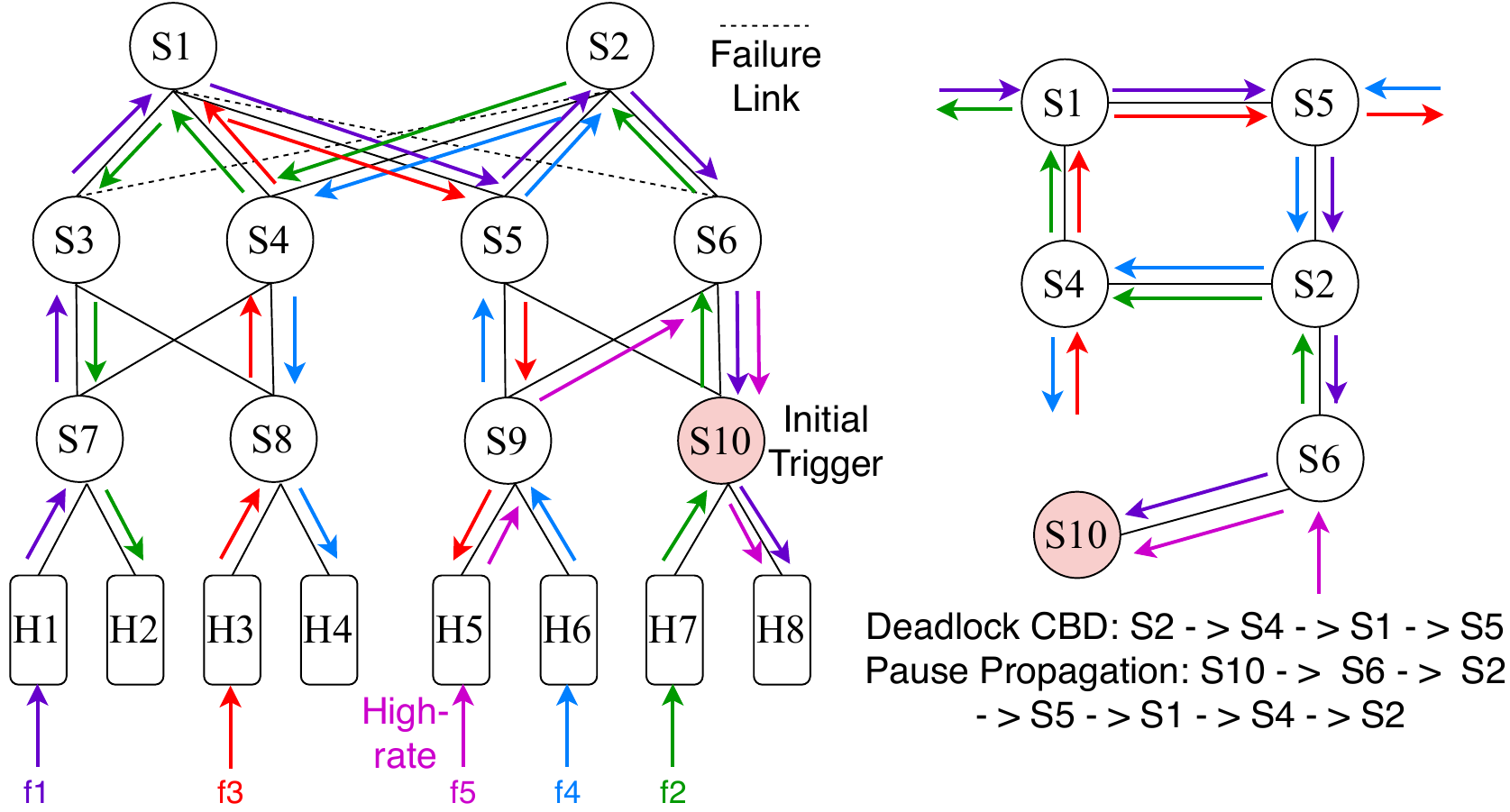}
    \vspace{-0.1in}
    \caption{Example of deadlock when initial trigger is out of loop}
    \label{fig:trigger out of loop}
    \vspace{-0.1in}
\end{figure}

\vspace{-1mm}
This condition could also be met even if there is no deadlock, as shown in Figure~\ref{fig:no deadlock}. $S7$ receives two pause frames with the same ID of $S1$ respectively from $S5$ and $S6$, however, no deadlock is formed. \sys can 
deal with the non-deadlock cases normally as well as the deadlock cases. The correctness of deadlock detection is guaranteed by the general primitives in Section~\ref{primitives}.

\textbf{Generic initiator selection:}
The middle switch receiving the same \{$S_{tri}$, $P_{tri}$\} from different ports is selected as the generic initiator to start a new process of deadlock detection. This process works in parallel with the previous process using the initial trigger as the generic initiator.

\textbf{Checking message propagation:}
As the middle switch is selected as the generic initiator, the \{$S_{gen-ini}$, $P_{gen-ini}$\} used for checking message in this case is the SwitchID and PortID of the middle switch, which we called $S_{middle}$ and $P_{middle}$. When the causality-loop detection is triggered, the middle switch has received two $Seq_{id}$ from different ports. The one received from the later port that has not stored the $Seq_{id}$ is selected for the subsequent deadlock detection. The middle switch generates a synthesized packet to carry the checking message \{$S_{middle}$, $P_{middle}$, $Seq_{id}$\}, which is forwarded along the causality-chain.

\textbf{Temporal consistency guarantee:}
Temporal consistency primitive is invoked by choosing the middle switch as the start of the temporal consistency check. 

Compared with the scenario that the initial trigger is on the loop, this case needs at most one additional round to detect the deadlock, which adds a small overhead.

\begin{figure}[t]
    \centering
    \includegraphics[width=0.3\textwidth]{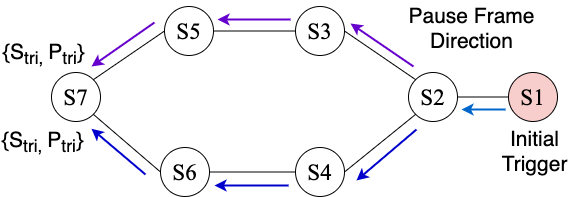}
    \vspace{-0.1in}
    \caption{A switch receives two pause frames with the same $S_{tri}$ from different ports, but no deadlock exists}
    \label{fig:no deadlock}
    \vspace{-0.1in}
\end{figure}

\section{Handling deadlock based on the initial trigger}
\label{section:discussion}

Upon detection of a deadlock, actions must be taken to break the deadlock and prevent future recurrence of the deadlock. Below we discuss several research directions. Details are left to future work.

\textbf{Breaking the deadlock:}
Packet dropping or temporary rerouting are common methods used to break a deadlock~\cite{martinez1997software, lopez1998very, shpiner2016unlocking}. 
However, rerouting is not always viable if not all flows have multiple paths. In addition, rerouting may create other deadlocks.
Breaking a deadlock by dropping packets is more practical and several fast synchronous draining approaches have been proposed~\cite{ramrakhyani2018synchronized}. The packet loss rate that can be tolerated by protocols like RoCEv2 is about 0.0001 ~\cite{zhu2015congestion}.
To break a deadlock, only a part of the buffered traffic needs to be dropped, thus the actual performance impact may be tolerable.

\textbf{Initial trigger handling:}
In addition to breaking the deadlock, further operations are needed to handle the initial trigger in order to prevent the deadlock from forming again.
If the initial trigger is a switch, a crucial step is to identify the heavy hitters which send a large amount of traffic and thus cause the congestion. Recent proposals for heavy hitter detection can be leveraged~\cite{sivaraman2017heavy, dixit2013impact}.
Once identified, further steps can be taken to limit the heavy hitters.

The initial trigger may also be a server when there is a flow control issue or a malfunctioning NIC. Due to limited memory resources in the NIC, a flow control issue may cause thousands of PFC pause frames to be sent per second from a server~\cite{guo2016rdma}.
In addition, bugs in the receiving pipeline of the NIC can cause the server to be unable to handle the received packets and continually send pause frames~\cite{zhu2015congestion, guo2016rdma}.
One approach for handling such an initial trigger server is to prevent the NIC from generating pause frames and disable lossless mode at the switch port. Alternatively, the connected switch can make use of dynamic buffer sharing, which improves the utilization of available buffer space and reduces PFC pause frame propagation.
\vspace{-3mm}
\section{Preliminary Results}
We prototype \sys with BMv2 software switches~\cite{bmv2} in the MiniNet environment for initial validation. The experiments are performed in the CloudLab platform~\cite{cloudlab}, each node has 4 cores and 32GB of RAM. We evaluate \sys for scenarios that the initial trigger is on the loop and out of the loop, using a $k=4$ fat-tree topology.

\textbf{Memory overhead:}
The memory overhead of \sys mainly comes from the port-based causality data structure and the maintained information of checking message. Both of them are determined by the number of ports in the network. In our experiments, each switch has $4$ ports and the memory overhead in this case is only $10^{-1}$KB. If we deploy \sys in a large-scale network with $10,000$ $64$-port switches, $14$-bit SwitchID and $6$-bit PortID are required to support the uniqueness requirement. The corresponding memory overhead is $1$KB per switch. Overall, the total memory requirement of \sys is quite small and it can be easily deployed in today's programmable data planes that usually have tens of MB of memory~\cite{tofino}.

\textbf{Detection effectiveness:}
We evaluate the effectiveness of \sys by creating deadlocks in two different scenarios---the initial trigger is on the loop and out of the loop respectively. 
We introduced two failed links so rerouting could finally lead to a deadlock loop. Traffic flows are generated across ToR switches and the resulting congestion could spread to the upstream switches and then to switches on the deadlock loop.
We repeat each scenario $10$ times. 
The average detection time for the \textit{on loop} scenario is $0.8$ms and $1.4$ms for the \textit{out of loop} scenario. In cases without causality-loop, no deadlock is declared by \sys. In cases that a causality-loop is formed but then one of the ports is immediately resumed, this situation is detected by the temporal consistency check without declaring any deadlock. There was no false positive or false negative.

\textbf{Benefits of resolving the initial trigger:}
To validate the benefits of resolving initial triggers, we simulate a misbehaving server that continually pauses the connected edge switch that leads to a deadlock. 
A baseline solution is to break the deadlock without resolving the initial trigger. We measure the normalized throughput of different flows and results are shown in Figure~\ref{fig:resolve}. Although the baseline can recover from the deadlock for a while, if the traffic pattern does not change, the deadlock will reappear. In contrast, \sys breaks the deadlock and resolves the initial trigger simultaneously, which prevents the recurrence of the deadlock.

\begin{figure}[t!]
    \centering
    \begin{subfigure}[b]{0.23\textwidth}
        \centering
        \includegraphics[width=\textwidth]{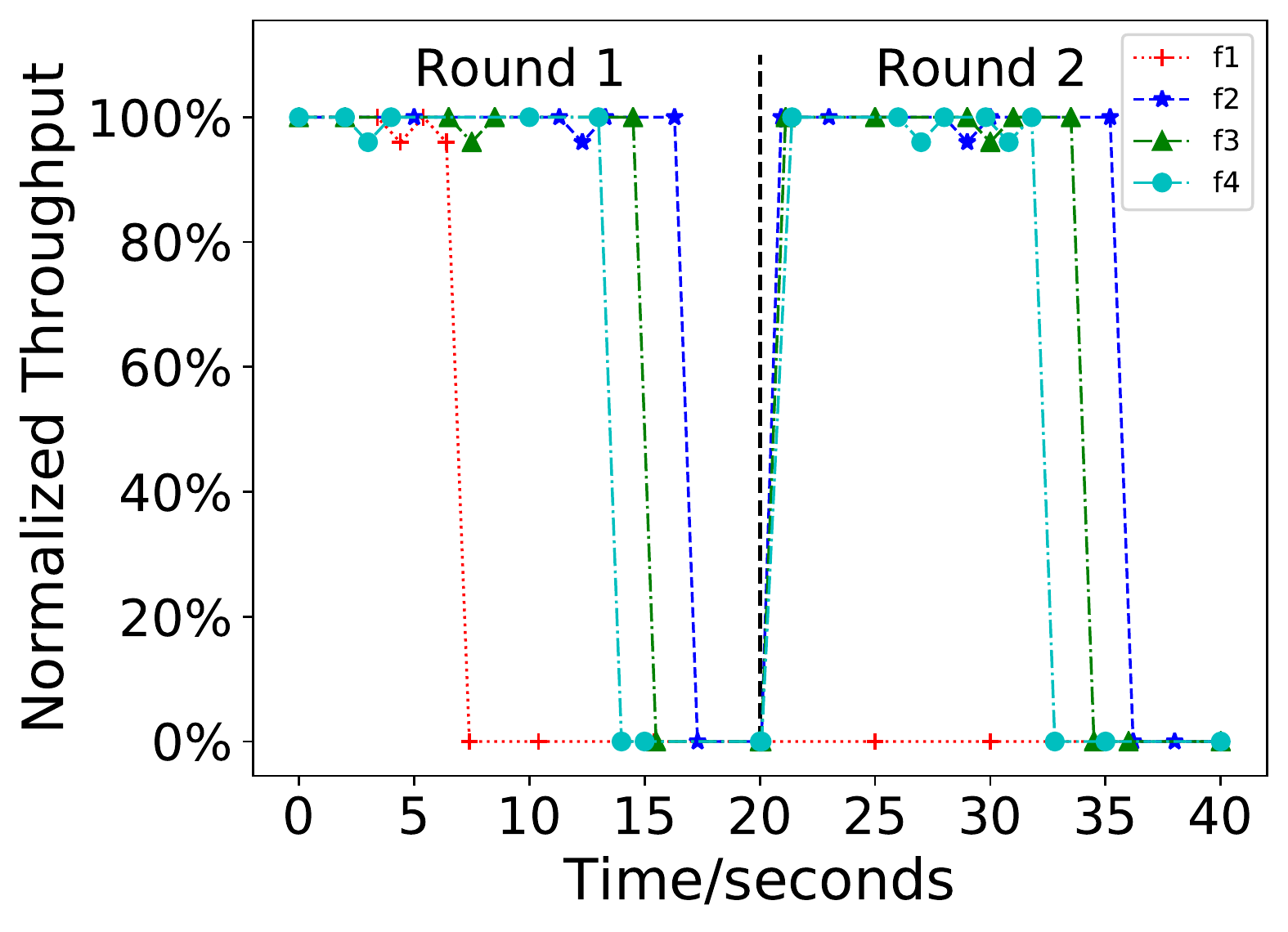}
        \caption{Without resolving initial trigger}
        \label{fig:resolve1}
     \end{subfigure}
     \hfill
    \begin{subfigure}[b]{0.23\textwidth}
        \centering
        \includegraphics[width=\textwidth]{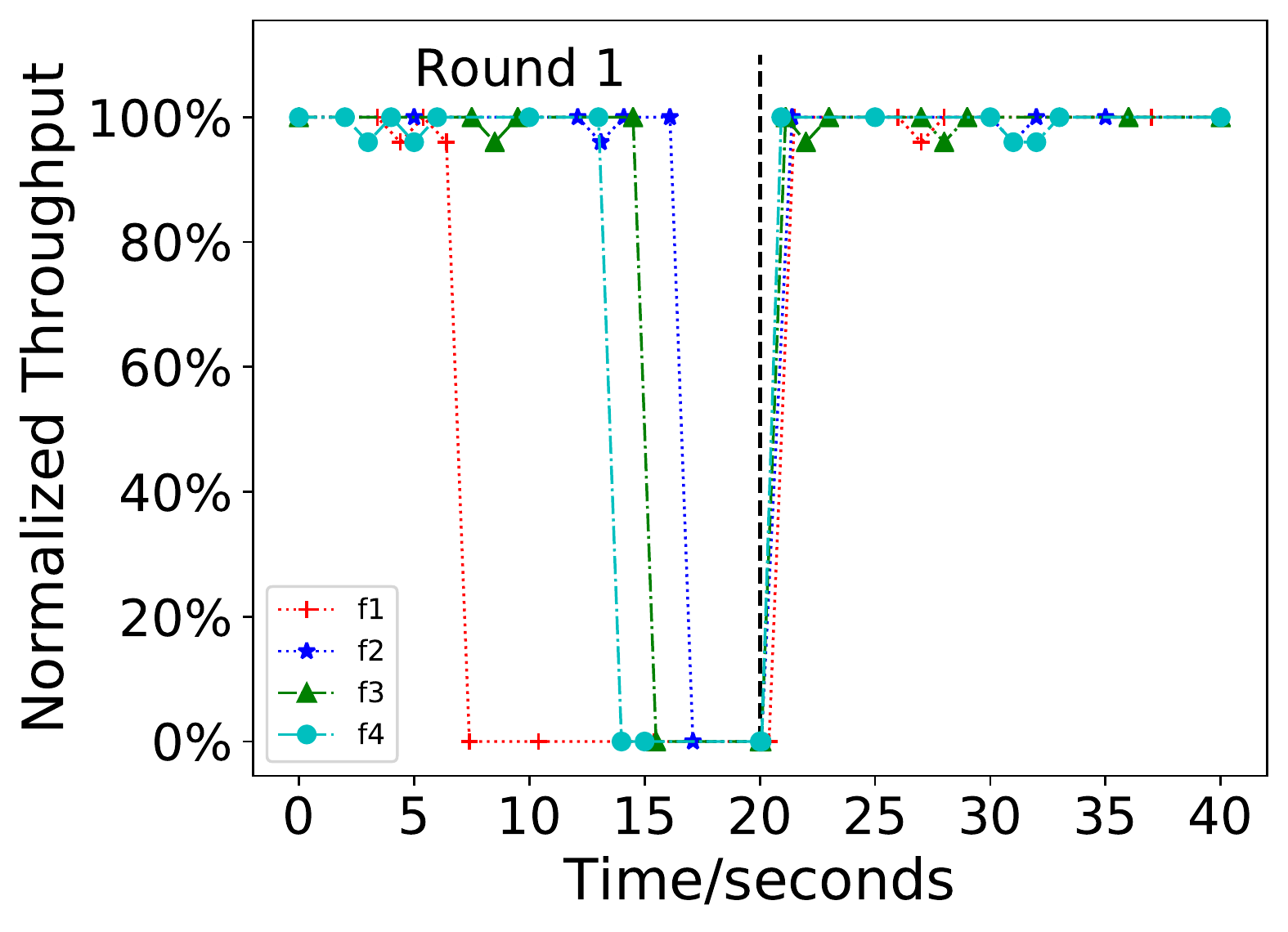}
        \caption{Resolving initial trigger}
        \label{fig:resolve2}
    \end{subfigure}
    \vspace{-2mm}
    \caption{Benefits of resolving initial trigger out of the loop 
    }
    \label{fig:resolve}
    \vspace{-0.2in}
\end{figure}
\vspace{-3mm}
\section{conclusion}
In this paper, we propose \sys to detect and resolve deadlocks in PFC networks. We identify the initial trigger to mitigate the recurrence of the same deadlock. The deadlock scenarios are analyzed and detected based on the location of the initial trigger. \sys can be implemented entirely in the data plane, which achieves low overhead and reacts quickly to deadlocks.

\section{Acknowledgement}
We would like to thank the BOLD Lab members for their useful feedback. This research is sponsored by the NSF under CNS-1718980, CNS-1801884, and CNS-1815525.

\bibliographystyle{abbrv} 
\begin{small}
\balance
\justifying  
\bibliography{main}

\begin{thebibliography}{10}

\bibitem{alizadeh2010data}
M.~Alizadeh, A.~Greenberg, D.~A. Maltz, J.~Padhye, P.~Patel, B.~Prabhakar,
  S.~Sengupta, and M.~Sridharan.
\newblock Data center tcp (dctcp).
\newblock In {\em Proceedings of the ACM SIGCOMM 2010 conference}, pages
  63--74, 2010.

\bibitem{anjan1995efficient}
K.~Anjan and T.~M. Pinkston.
\newblock An efficient, fully adaptive deadlock recovery scheme: Disha.
\newblock In {\em Proceedings of the 22nd annual international symposium on
  Computer architecture}, pages 201--210, 1995.

\bibitem{arashloo2016snap}
M.~T. Arashloo, Y.~Koral, M.~Greenberg, J.~Rexford, and D.~Walker.
\newblock Snap: Stateful network-wide abstractions for packet processing.
\newblock In {\em Proceedings of the 2016 ACM SIGCOMM Conference}, pages
  29--43, 2016.

\bibitem{bansal2012openradio}
M.~Bansal, J.~Mehlman, S.~Katti, and P.~Levis.
\newblock Openradio: a programmable wireless dataplane.
\newblock In {\em Proceedings of the first workshop on Hot topics in software
  defined networks}, pages 109--114, 2012.

\bibitem{tofino}
Barefoot.
\newblock Tofino: World’s fastest p4-programmable ethernet switch asics.
\newblock https://www.barefootnetworks.com/products/brief-tofino/.

\bibitem{beckett2017general}
R.~Beckett, A.~Gupta, R.~Mahajan, and D.~Walker.
\newblock A general approach to network configuration verification.
\newblock In {\em Proceedings of the Conference of the ACM Special Interest
  Group on Data Communication}, pages 155--168, 2017.

\bibitem{benson2010network}
T.~Benson, A.~Akella, and D.~A. Maltz.
\newblock Network traffic characteristics of data centers in the wild.
\newblock In {\em Proceedings of the 10th ACM SIGCOMM conference on Internet
  measurement}, pages 267--280, 2010.

\bibitem{birkner2020config2spec}
R.~Birkner, D.~Drachsler-Cohen, L.~Vanbever, and M.~Vechev.
\newblock Config2spec: Mining network specifications from network
  configurations.
\newblock In {\em Proceedings of 17th USENIX Symposium on Networked Systems
  Design and Implementation (NSDI'20)}, 2020.

\bibitem{blazewicz1994optimal}
J.~Blazewicz, D.~P. Bovet, J.~Brzezinski, G.~Gambosi, and M.~Talamo.
\newblock Optimal centralized algorithms for store-and-forward deadlock
  avoidance.
\newblock {\em IEEE transactions on computers}, 43(11):1333--1338, 1994.

\bibitem{cheng2020re}
W.~Cheng, K.~Qian, W.~Jiang, T.~Zhang, and F.~Ren.
\newblock Re-architecting congestion management in lossless ethernet.
\newblock In {\em 17th USENIX Symposium on Networked Systems Design and
  Implementation (NSDI 20)}, pages 19--36, 2020.

\bibitem{dally1988deadlock}
W.~J. Dally and C.~L. Seitz.
\newblock Deadlock-free message routing in multiprocessor interconnection
  networks.
\newblock 1988.

\bibitem{dixit2013impact}
A.~Dixit, P.~Prakash, Y.~C. Hu, and R.~R. Kompella.
\newblock On the impact of packet spraying in data center networks.
\newblock In {\em 2013 Proceedings IEEE INFOCOM}, pages 2130--2138. IEEE, 2013.

\bibitem{domke2011deadlock}
J.~Domke, T.~Hoefler, and W.~E. Nagel.
\newblock Deadlock-free oblivious routing for arbitrary topologies.
\newblock In {\em 2011 IEEE International Parallel \& Distributed Processing
  Symposium}, pages 616--627. IEEE, 2011.

\bibitem{duato2001general}
J.~Duato and T.~M. Pinkston.
\newblock A general theory for deadlock-free adaptive routing using a mixed set
  of resources.
\newblock {\em IEEE Transactions on Parallel and Distributed Systems},
  12(12):1219--1235, 2001.

\bibitem{cloudlab}
D.~Duplyakin, R.~Ricci, A.~Maricq, G.~Wong, J.~Duerig, E.~Eide, L.~Stoller,
  M.~Hibler, D.~Johnson, K.~Webb, A.~Akella, K.~Wang, G.~Ricart, L.~Landweber,
  C.~Elliott, M.~Zink, E.~Cecchet, S.~Kar, and P.~Mishra.
\newblock The design and operation of cloudlab.
\newblock In {\em Proceedings of the USENIX Annual Technical Conference (ATC)},
  pages 1--14, 2019.

\bibitem{fogel2015general}
A.~Fogel, S.~Fung, L.~Pedrosa, M.~Walraed-Sullivan, R.~Govindan, R.~Mahajan,
  and T.~Millstein.
\newblock A general approach to network configuration analysis.
\newblock In {\em 12th USENIX Symposium on Networked Systems Design and
  Implementation (NSDI 15)}, pages 469--483, 2015.

\bibitem{forster2016consistent}
K.-T. F{\"o}rster, R.~Mahajan, and R.~Wattenhofer.
\newblock Consistent updates in software defined networks: On dependencies,
  loop freedom, and blackholes.
\newblock In {\em 2016 IFIP Networking Conference (IFIP Networking) and
  Workshops}, pages 1--9. IEEE, 2016.

\bibitem{geng2019p4qcn}
J.~Geng, J.~Yan, and Y.~Zhang.
\newblock P4qcn: Congestion control using p4-capable device in data center
  networks.
\newblock {\em Electronics}, 8(3):280, 2019.

\bibitem{guo2016rdma}
C.~Guo, H.~Wu, Z.~Deng, G.~Soni, J.~Ye, J.~Padhye, and M.~Lipshteyn.
\newblock Rdma over commodity ethernet at scale.
\newblock In {\em Proceedings of the 2016 ACM SIGCOMM Conference}, pages
  202--215, 2016.

\bibitem{halperin2011augmenting}
D.~Halperin, S.~Kandula, J.~Padhye, P.~Bahl, and D.~Wetherall.
\newblock Augmenting data center networks with multi-gigabit wireless links.
\newblock In {\em Proceedings of the ACM SIGCOMM 2011 conference}, pages
  38--49, 2011.

\bibitem{hancock2016hyper4}
D.~Hancock and J.~Van~der Merwe.
\newblock Hyper4: Using p4 to virtualize the programmable data plane.
\newblock In {\em Proceedings of the 12th International on Conference on
  emerging Networking EXperiments and Technologies}, pages 35--49, 2016.

\bibitem{hu2016deadlocks}
S.~Hu, Y.~Zhu, P.~Cheng, C.~Guo, K.~Tan, J.~Padhye, and K.~Chen.
\newblock Deadlocks in datacenter networks: Why do they form, and how to avoid
  them.
\newblock In {\em Proceedings of the 15th ACM Workshop on Hot Topics in
  Networks}, pages 92--98, 2016.

\bibitem{hu2017tagger}
S.~Hu, Y.~Zhu, P.~Cheng, C.~Guo, K.~Tan, J.~Padhye, and K.~Chen.
\newblock Tagger: Practical pfc deadlock prevention in data center networks.
\newblock In {\em Proceedings of the 13th International Conference on emerging
  Networking EXperiments and Technologies}, pages 451--463, 2017.

\bibitem{ieee}
IEEE.
\newblock Ieee 802.1 qbb - priority-based flow control.
\newblock https://1.ieee802.org/dcb/802-1qbb/, 2010.

\bibitem{jin2014dynamic}
X.~Jin, H.~H. Liu, R.~Gandhi, S.~Kandula, R.~Mahajan, M.~Zhang, J.~Rexford, and
  R.~Wattenhofer.
\newblock Dynamic scheduling of network updates.
\newblock {\em ACM SIGCOMM Computer Communication Review}, 44(4):539--550,
  2014.

\bibitem{kakarla2020finding}
S.~K.~R. Kakarla, A.~Tang, R.~Beckett, K.~Jayaraman, T.~Millstein, Y.~Tamir,
  and G.~Varghese.
\newblock Finding network misconfigurations by automatic template inference.
\newblock In {\em 17th USENIX Symposium on Networked Systems Design and
  Implementation (NSDI 20)}, pages 999--1013, 2020.

\bibitem{lee2005prevention}
D.~Lee, S.~J. Golestani, and M.~J. Karol.
\newblock Prevention of deadlocks and livelocks in lossless, backpressured
  packet networks, Feb.~22 2005.
\newblock US Patent 6,859,435.

\bibitem{li2019hpcc}
Y.~Li, R.~Miao, H.~H. Liu, Y.~Zhuang, F.~Feng, L.~Tang, Z.~Cao, M.~Zhang,
  F.~Kelly, M.~Alizadeh, et~al.
\newblock Hpcc: high precision congestion control.
\newblock In {\em Proceedings of the ACM Special Interest Group on Data
  Communication}, pages 44--58. 2019.

\bibitem{liu2013f10}
V.~Liu, D.~Halperin, A.~Krishnamurthy, and T.~Anderson.
\newblock F10: A fault-tolerant engineered network.
\newblock In {\em Presented as part of the 10th USENIX Symposium on Networked
  Systems Design and Implementation (NSDI 13)}, pages 399--412, 2013.

\bibitem{lopez1998very}
P.~Lopez, J.~M. Mart{\'\i}nez, and J.~Duato.
\newblock A very efficient distributed deadlock detection mechanism for
  wormhole networks.
\newblock In {\em Proceedings 1998 Fourth International Symposium on
  High-Performance Computer Architecture}, pages 57--66. IEEE, 1998.

\bibitem{lou2020understanding}
C.~Lou, P.~Huang, and S.~Smith.
\newblock Understanding, detecting and localizing partial failures in large
  system software.
\newblock In {\em 17th USENIX Symposium on Networked Systems Design and
  Implementation (NSDI 20)}, pages 559--574, 2020.

\bibitem{martinez1997software}
J.~M. Mart{\'\i}nez, P.~Lopez, J.~Duato, and T.~M. Pinkston.
\newblock Software-based deadlock recovery technique for true fully adaptive
  routing in wormhole networks.
\newblock In {\em Proceedings of the 1997 International Conference on Parallel
  Processing}, pages 182--189. IEEE, 1997.

\bibitem{mittalgood}
R.~Mittal.
\newblock Good performance using common network infrastructure.

\bibitem{mittal2018towards}
R.~Mittal.
\newblock {\em Towards a More Stable Network Infrastructure}.
\newblock University of California, Berkeley, 2018.

\bibitem{mittal2015timely}
R.~Mittal, V.~T. Lam, N.~Dukkipati, E.~Blem, H.~Wassel, M.~Ghobadi, A.~Vahdat,
  Y.~Wang, D.~Wetherall, and D.~Zats.
\newblock Timely: Rtt-based congestion control for the datacenter.
\newblock {\em ACM SIGCOMM Computer Communication Review}, 45(4):537--550,
  2015.

\bibitem{mittal2018revisiting}
R.~Mittal, A.~Shpiner, A.~Panda, E.~Zahavi, A.~Krishnamurthy, S.~Ratnasamy, and
  S.~Shenker.
\newblock Revisiting network support for rdma.
\newblock In {\em Proceedings of the 2018 Conference of the ACM Special
  Interest Group on Data Communication}, pages 313--326, 2018.

\bibitem{bmv2}
P4lang.
\newblock P4 behavioral model.
\newblock https://github.com/p4lang/behavioral-model.

\bibitem{park1996generic}
H.~Park and D.~P. Agrawal.
\newblock Generic methodologies for deadlock-free routing.
\newblock In {\em Proceedings of International Conference on Parallel
  Processing}, pages 638--643. IEEE, 1996.

\bibitem{qian2019gentle}
K.~Qian, W.~Cheng, T.~Zhang, and F.~Ren.
\newblock Gentle flow control: avoiding deadlock in lossless networks.
\newblock In {\em Proceedings of the ACM Special Interest Group on Data
  Communication}, pages 75--89. 2019.

\bibitem{ramrakhyani2018synchronized}
A.~Ramrakhyani, P.~V. Gratz, and T.~Krishna.
\newblock Synchronized progress in interconnection networks (spin): A new
  theory for deadlock freedom.
\newblock In {\em 2018 ACM/IEEE 45th Annual International Symposium on Computer
  Architecture (ISCA)}, pages 699--711. IEEE, 2018.

\bibitem{sancho2004effective}
J.~C. Sancho, A.~Robles, and J.~Duato.
\newblock An effective methodology to improve the performance of the up*/down*
  routing algorithm.
\newblock {\em IEEE Transactions on Parallel and Distributed Systems},
  15(8):740--754, 2004.

\bibitem{shpiner2016unlocking}
A.~Shpiner, E.~Zahavi, V.~Zdornov, T.~Anker, and M.~Kadosh.
\newblock Unlocking credit loop deadlocks.
\newblock In {\em Proceedings of the 15th ACM Workshop on Hot Topics in
  Networks}, pages 85--91, 2016.

\bibitem{sivaraman2017heavy}
V.~Sivaraman, S.~Narayana, O.~Rottenstreich, S.~Muthukrishnan, and J.~Rexford.
\newblock Heavy-hitter detection entirely in the data plane.
\newblock In {\em Proceedings of the Symposium on SDN Research}, pages
  164--176, 2017.

\bibitem{skeie2002layered}
T.~Skeie, O.~Lysne, and I.~Theiss.
\newblock Layered shortest path (lash) routing in irregular system area
  networks.
\newblock In {\em Proceedings 16th International Parallel and Distributed
  Processing Symposium. IPDPS 2002}, pages 8--pp. Citeseer, 2002.

\bibitem{stephens2016deadlock}
B.~Stephens and A.~L. Cox.
\newblock Deadlock-free local fast failover for arbitrary data center networks.
\newblock In {\em IEEE INFOCOM 2016-The 35th Annual IEEE International
  Conference on Computer Communications}, pages 1--9. IEEE, 2016.

\bibitem{stephens2014practical}
B.~Stephens, A.~L. Cox, A.~Singla, J.~Carter, C.~Dixon, and W.~Felter.
\newblock Practical dcb for improved data center networks.
\newblock In {\em IEEE INFOCOM 2014-IEEE Conference on Computer
  Communications}, pages 1824--1832. IEEE, 2014.

\bibitem{tan2019netbouncer}
C.~Tan, Z.~Jin, C.~Guo, T.~Zhang, H.~Wu, K.~Deng, D.~Bi, and D.~Xiang.
\newblock Netbouncer: active device and link failure localization in data
  center networks.
\newblock In {\em 16th USENIX Symposium on Networked Systems Design and
  Implementation (NSDI 19)}, pages 599--614, 2019.

\bibitem{underwood2011unified}
K.~D. Underwood and E.~Borch.
\newblock A unified algorithm for both randomized deterministic and adaptive
  routing in torus networks.
\newblock In {\em 2011 IEEE International Symposium on Parallel and Distributed
  Processing Workshops and Phd Forum}, pages 723--732. IEEE, 2011.

\bibitem{wu2003fault}
J.~Wu.
\newblock A fault-tolerant and deadlock-free routing protocol in 2d meshes
  based on odd-even turn model.
\newblock {\em IEEE Transactions on Computers}, 52(9):1154--1169, 2003.

\bibitem{wu2012netpilot}
X.~Wu, D.~Turner, C.-C. Chen, D.~A. Maltz, X.~Yang, L.~Yuan, and M.~Zhang.
\newblock Netpilot: automating datacenter network failure mitigation.
\newblock In {\em Proceedings of the ACM SIGCOMM 2012 conference on
  Applications, technologies, architectures, and protocols for computer
  communication}, pages 419--430, 2012.

\bibitem{zats2012detail}
D.~Zats, T.~Das, P.~Mohan, D.~Borthakur, and R.~Katz.
\newblock Detail: reducing the flow completion time tail in datacenter
  networks.
\newblock In {\em Proceedings of the ACM SIGCOMM 2012 conference on
  Applications, technologies, architectures, and protocols for computer
  communication}, pages 139--150, 2012.

\bibitem{zats2015fastlane}
D.~Zats, A.~P. Iyer, G.~Ananthanarayanan, R.~Agarwal, R.~Katz, I.~Stoica, and
  A.~Vahdat.
\newblock Fastlane: making short flows shorter with agile drop notification.
\newblock In {\em Proceedings of the Sixth ACM Symposium on Cloud Computing},
  pages 84--96, 2015.

\bibitem{zhu2015congestion}
Y.~Zhu, H.~Eran, D.~Firestone, C.~Guo, M.~Lipshteyn, Y.~Liron, J.~Padhye,
  S.~Raindel, M.~H. Yahia, and M.~Zhang.
\newblock Congestion control for large-scale rdma deployments.
\newblock {\em ACM SIGCOMM Computer Communication Review}, 45(4):523--536,
  2015.

\bibitem{zhu2015packet}
Y.~Zhu, N.~Kang, J.~Cao, A.~Greenberg, G.~Lu, R.~Mahajan, D.~Maltz, L.~Yuan,
  M.~Zhang, B.~Y. Zhao, et~al.
\newblock Packet-level telemetry in large datacenter networks.
\newblock In {\em Proceedings of the 2015 ACM Conference on Special Interest
  Group on Data Communication}, pages 479--491, 2015.

\end{thebibliography}
\end{small}

\end{document}